\documentclass[12pt]{article}
\usepackage[letterpaper,top=1.5cm, bottom=1.6cm, left=1.6cm, right=1.6cm]{geometry}
 
\usepackage{bbm}
\usepackage[normalem]{ulem}
 \usepackage{amssymb}
\usepackage{amsthm}
\usepackage{bm}
\usepackage{mathrsfs}
\usepackage{amsmath}
\usepackage{dsfont}
\usepackage{mathtools}
\usepackage{graphicx}
\usepackage[font=normalsize]{caption}
\usepackage{float}
\usepackage{color}
\usepackage{overpic}
\usepackage{graphicx}
\usepackage{mathrsfs}
 \usepackage{enumerate}
\usepackage{mathtools}

\usepackage{soul}
\usepackage[normalem]{ulem}

\theoremstyle{plain}
\newtheorem{theorem}{Theorem}[section]
\newtheorem{lemma}[theorem]{Lemma}
\newtheorem{proposition}[theorem]{Proposition}

\theoremstyle{definition}
\newtheorem{definition}[theorem]{Definition}
\newtheorem{remark}[theorem]{Remark}
\newtheorem{assumption}[theorem]{Assumption}

\begin{document}

\title{Nash equilibrium for risk-averse investors in a market impact game with transient price impact}

\author{Xiangge Luo\thanks{Seminar for Statistics, ETH Z\"urich, CH-8092 Z\"urich, Switzerland. Email: {\tt luox@student.ethz.ch}}\and \setcounter{footnote}{6} Alexander Schied\thanks{Department of Statistics and Actuarial Science, University of Waterloo,   Waterloo, ON N2L3G1, Canada. Email: {\tt aschied@uwaterloo.ca}\hfill\break
The authors gratefully acknowledge support from the
 Natural Sciences and Engineering Research Council of Canada through grants RGPIN-2017-04054 and USRA-515434-2017.}}
\date{\normalsize First version: July 10, 2018\\
This version: May 19, 2019}

\maketitle

\begin{abstract}
We consider a market impact game for $n$ risk-averse agents that are competing in a market model with linear transient price impact and additional transaction costs. For both finite and infinite time horizons, the agents  aim to minimize a mean-variance functional of their costs or to maximize the expected exponential utility of their revenues. We give explicit representations for corresponding Nash equilibria and prove uniqueness in the case of mean-variance optimization. A qualitative analysis of these Nash equilibria is conducted by means of numerical analysis. 
\end{abstract}


\section{Introduction}
\label{S:1}

In a market impact game, financial agents compete with each other in a market framework where each trade creates price impact. 
Early papers on this subject, such as  Brunnermeier and Pedersen~\cite{BrunnermeierPedersen}, Carlin et al.~\cite{Carlinetal}, and Sch\"oneborn and Schied~\cite{SchoenebornSchied}, consider risk-neutral agents that are active in a linear Almgren--Chriss market impact model. Already in this relatively simple setup, interesting effects appear, such as transitions in the predatory or cooperative behavior of agents. Extensions to risk-averse agents in the Almgren--Chriss framework were given, e.g., in~\cite{CarmonaYang,SchiedZhangCARA}. For further developments in the literature on market impact games, we refer to \cite{Carmonabook,Moallemietal,CasgrainJaimungal}.

In this paper, we consider risk-averse agents that are active in a discrete-time model with linear transient price impact. For single-agent optimization problems, such price-impact models were introduced by Obizhaeva and Wang~\cite{ow} and later further developed, e.g.,  in~\cite{ASS,Gatheral}. A market impact game with two risk-neutral agents was first considered by Sch\"oneborn~\cite{Schoeneborn}, who observed that equilibrium strategies may exhibit strong oscillations between buy and sell trades if trading speed is sufficiently high. This situation was further investigated in~\cite{SchiedStrehleZhang,SchiedZhangHotPotato}, where the model was also enhanced by introducing additional quadratic transaction costs, whose strength is parameterized by a number $\theta\ge0$.  It was shown in particular that  there exists an explicitly given critical value $\theta^*> 0$ such that the equilibrium strategies show at least some oscillations for $\theta<\theta^*$, whereas all oscillations disappear for $\theta\ge\theta^*$.

In~\cite{SchiedZhangHotPotato}, only two competing risk-neutral agents are considered. The main goal of the present paper is to extend the results and observations from~\cite{SchiedZhangHotPotato} to a more flexible setting, in which an arbitrary (but finite) number of agents optimize their strategies under risk aversion. More precisely, the agents either minimize a mean-variance functional of the trading costs over deterministic strategies or they maximize the expected CARA utility of their revenues over adaptive strategies. We show that both problems admit an identical Nash equilibrium, which is given in explicit form and which is unique in the case of mean-variance optimization. More precisely, the equilibrium strategies arise as linear combinations of two extreme base strategies $\bm v$ and $\bm w$. The first, $\bm v$, is the normalized common strategy of all players if each player has the same initial position.  The second, $\bm w$,  is the normalized common strategy of all players if the initial positions of all players add up to zero.

Then we use numerical analysis of the equilibrium strategies to determine numerically the critical threshold for the transaction costs above which all oscillations cease. In contrast to the risk-neutral two-player case studied in \cite{SchiedZhangHotPotato}, we now observe two different thresholds for $\bm v$ and $\bm w$. Moreover, the threshold for $\bm v$ will not depend on the risk-aversion but on the number of players. By contrast, the threshold for $\bm w$ will not depend on the number of players but on the risk aversion.

If agents exhibit strictly positive risk aversion, it is possible to study the market impact game with an infinite time horizon. This question is interesting when one does not want to impose an externally given time horizon and instead aims at an intrinsic derivation of a trading horizon. We show that such an infinite-horizon market impact game admits a Nash equilibrium in case $\theta$ is equal to the critical value $\theta^*$, which was determined numerically for finite time horizons. If $\theta\neq\theta^*$, a Nash equilibrium may not exist. 

This conjectured nonexistence is a consequence of the idealization of admitting infinitely many trades, an idealization that also in the context of continuous-time models has turned out to be not as innocent as one might initially hope. Specifically, it was shown in \cite{SchiedStrehleZhang} that in the case of two risk-neutral agents, a nontrivial Nash equilibrium can only exist if $\theta=\theta^*$. This negative result has motivated Strehle \cite{Strehle} to include in the cost functional an additional penalization of the derivatives of the continuous-time strategies. This additional term regularizes the admissible strategies so that Nash equilibria exist in general.

This paper is organized as follows. In Section~\ref{finite horizon section}, we present the setup and state all results for a finite time horizon. Section~\ref{sec:inf}
 contains our discussion of the market impact game with infinite time horizon. All proofs are given in Section~\ref{proofs section}. 

\section{Main results}
\label{S:2}

\subsection{Finite time horizon}\label{finite horizon section}

We consider an $n$-agent extension of the discrete-time market impact model with linear transient price impact that was studied, e.g., in~\cite{AFS1,ASS,ow,SchiedZhangHotPotato,Schoeneborn}. This model is sometimes also called the discrete-time  linear propagator model, and we refer to~\cite{GatheralSchiedSurvey} for a discussion and further background. 

Suppose that $n$ financial agents are active in a market impact model for one risky asset. As commonly assumed in the market impact literature, the unaffected price process $S^0=(S^0_t)_{t\ge0}$ will be a square-integrable and right-continuous martingale on the filtered probability space $(\Omega,\mathscr{F},(\mathscr{F}_t)_{t\ge0},\mathbb P)$. An important  special case will be the Bachelier model of the form
\begin{equation}\label{Bachelier model}
S^0_t=S_0+\sigma B_t,\qquad t\ge0,
\end{equation}
for constants $S_0,\sigma>0$ and a standard Brownian motion $B$.
All agents trade at a finite number of times $0\le t_0<t_1<\cdots<t_N$. The trading strategy of agent $i$ will be a vector $\bm\xi_i=(\xi_{i,0},\dots,\xi_{i,N})^\top$ where $\xi_{i,k}$ represents the number of shares sold at time $t_k$. That is,  $\xi_{i,k} > 0$ represents a  sell order and $\xi_{i,k} < 0$ means a buy order. The matrix of all strategies is denoted by $\Xi = [\bm\xi_1, \dots, \bm\xi_n]$.

 When all the agents apply their strategies, the asset price is given by
\begin{equation}
\label{price}
S^{\Xi}_t = S^0_t - \sum_{t_k<t} \Big[G(t-t_k) \sum_{i=1}^n \xi_{i,k} \Big],
\end{equation}
where $G: \mathbb{R}_+ \rightarrow \mathbb{R}_+$ is called the \emph{decay kernel}.  The quantity $G(t-t_k)$ describes the time-$t$ price impact of a unit transaction made at time $t_k\le t$. When  agent $j$ first places an order $\xi_{j,k} >0$ at time $t_k$, the asset price is moved linearly from $S^{\Xi}_{t_k}$ to $S^{\Xi}_{t_k+} := S^{\Xi}_{t_k} - G(0) \xi_{j,k}$. The liquidation cost for agent $j$ is thus:
\begin{equation*}
-\frac{1}{2}(S^{\Xi}_{t_k+} + S^{\Xi}_{t_k})\xi_{j,k} = \frac{G(0)}{2}\xi_{j,k}^2 - S^{\Xi}_{t_k}\xi_{j,k}.
\end{equation*}
Suppose that immediately after agent $j$, another agent $i$ places an order $\xi_{i,k}>0$. The liquidation cost for agent $i$ is the following:
\begin{equation}\label{single price impact costs}
-\frac{1}{2}(S^{\Xi}_{t_k+} + S^{\Xi}_{t_k+} - G(0) \xi_{i,k})\xi_{i,k} = \frac{G(0)}{2}\xi_{i,k}^2 - S^{\Xi}_{t_k}\xi_{i,k} + G(0)\xi_{j,k}\xi_{i,k},
\end{equation}
where $G(0)\xi_{j,k}\xi_{i,k}$ is an additional cost term due to the latency in execution time. On average, fifty percent of times, the order of agent $j$ will be executed before the order of agent $i$. The latency costs for agent $i$ at time $t_k$ will thus be of the form
$$\frac12 G(0)\sum_{j\neq i}\xi_{i,k}\xi_{j,k}.
$$
In addition to the execution costs described above, we follow~\cite{SchiedStrehleZhang,SchiedZhangHotPotato} in assuming quadratic transaction costs $\theta\xi_{i,k}^2$ with $\theta \geq 0$. One of our goals will be to analyze the qualitative effects these transaction costs will have on optimal strategies. Such quadratic transaction costs are often used to model   \lq\lq slippage" arising from various costs incurred by a transaction (see~\cite{AlmgrenChriss2,BertsimasLo} and~\cite[Section 2.2]{Gatheral}) or a transaction tax (see \cite{SchiedStrehleZhang,SchiedZhangHotPotato}). As discussed by Strehle \cite[p. 5]{Strehle}, these transaction costs should not be understood as resulting from temporary price impact, as all costs arising from price impact are already contained in \eqref{single price impact costs}. Moreover, one can argue as in Proposition 2.6 of~\cite{SchiedZhangHotPotato} to see that  our quadratic transaction cost function can be replaced by proportional transaction costs in a neighborhood of the origin without affecting the Nash equilibrium we are going to derive. Since the main difference of quadratic and proportional transaction costs is their behavior at the origin, it is therefore highly plausible that similar results as obtained in the following sections for quadratic transaction costs might also hold for proportional transaction costs. The previous discussion thus motivates the following definition.

\begin{definition}
\label{def:cost}
Given  a time grid $\mathbb{T} = \{t_0, t_1, \dots, t_N\}$, the \emph{execution costs} of the strategy $\bm{\xi}_i $ given all other strategies $\bm{\xi}_j  $ with $j \neq i$ are defined as
\begin{equation}\label{cost functional}
 \mathscr{C}_{\mathbb{T}}(\bm{\xi}_i  | \bm{\xi}_{-i} )= 
  \sum^N_{k=0}\Big[\frac{G(0)}{2}\xi_{i,k}^2 - S^{\Xi}_{t_k}\xi_{i,k}  + \frac{G(0)}2\sum_{j\neq i}\xi_{i,k}\xi_{j,k}  + \theta\xi_{i,k}^2\Big],
\end{equation}
where
$\bm{\xi}_{-i}=[\bm{\xi}_{1}, \dots, \bm{\xi}_{i-1}, \bm{\xi}_{i+1}, \dots, \bm{\xi}_{n}]$.
\end{definition}

In the sequel, we will suppose that agent $i$ has an initial position of $X_i\in\mathbb R$ shares   and is constrained to hold a zero terminal position by the end of the trading day. It is often assumed~\cite{ow,SchiedZhangHotPotato} that  agents aim  to minimize the expected costs  
over the following class 
 of strategies,
$$\mathscr{X}(X_i,\mathbb T)=\Big\{\bm\xi=(\xi_0,\dots\xi_N)\,\Big|\,\xi_i\text{ is $\mathscr{F}_{t_i}$-measurable and bounded and }\sum_{i=0}^N\xi_i=X_i\Big\}.
$$
 In practice, however, it is also popular to incorporate the agents' risk aversion  and to  optimize the following mean-variance functional of the trading costs,
\begin{equation}
\label{eq:mvf}
\mathrm{MV}_{\gamma}(\bm{\xi}_i  | \bm{\xi}_{-i}) = \mathbb{E}[\mathscr{C}_{\mathbb{T}}(\bm{\xi}_i  | \bm{\xi}_{-i})] + \frac{\gamma}{2}\mathrm{Var}  [\mathscr{C}_{\mathbb{T}}(\bm{\xi}_i  | \bm{\xi}_{-i})].
\end{equation}
Here, $\gamma\ge0$ is a risk-aversion parameter. For $\gamma>0$, the mean-variance functional \eqref{eq:mvf} is typically only time-consistent if   strategies are deterministic; see, e.g.,~\cite{AlmgrenChriss2,LorenzAlmgren}.   Therefore, its minimization is usually restricted to the class of deterministic strategies in $\mathscr{X}(X_i,\mathbb T)$, which we denote by
$$\mathscr{X}_{\mathrm{det}}(X_i,\mathbb T)=\Big\{\bm\xi\in \mathscr{X}(X_i,\mathbb T)\,\Big|\,\text{$\bm\xi$ is deterministic}\Big\}=\Big\{\bm\xi\in\mathbb R^{|\mathbb T|}\,\Big|\,\bm1^\top\bm\xi=X_i\Big\},
$$
for $\bm 1=(1,\dots,1)^\top \in \mathbb{R}^{N+1}$. It can also make sense to maximize the expected utility of the revenues, which are nothing else than the negative costs. Here, we will use the following utility functional,
$$U_\gamma(\bm\xi_i|\bm\xi_{-i}):=\mathbb E[u_\gamma(-\mathscr{C}_{\mathbb{T}}(\bm{\xi}_i  | \bm{\xi}_{-i}))],
$$
where $u_\gamma(x)$ is the following exponential -- or CARA -- utility function,
$$u_\gamma(x)=\begin{cases}\frac1\gamma(1-e^{-\gamma x})&\text{if $\gamma>0$,}\\
-x&\text{if $\gamma=0$.}
\end{cases}
$$
Due to the time consistency of the expected utility functional, we can consider its maximization over all adapted strategies from the class $\mathscr{X}(X_i,\mathbb T)$.
Moreover, as, e.g.,  in~\cite{Carlinetal,SchiedStrehleZhang,SchiedZhangCARA,SchiedZhangHotPotato,Strehle}, we assume henceforth that each agent has full information about the strategies used by the other agents.  

\begin{definition}
\label{def:nash}
Suppose there are $n$ agents with initial inventories $X_1, \dots, X_n \in \mathbb{R}$ and risk aversion parameter $\gamma \geq 0$ and that $\mathbb{T} := \{t_0, t_1, \dots, t_N\}$ is a fixed time grid.
\begin{enumerate}[{\rm(a)}]

 \item A \emph{Nash equilibrium for mean-variance optimization} is  a collection of strategies $(\bm{\xi}_1^*, \dots, \bm{\xi}_n^*) \in \mathscr{X}_{\mathrm{det}} (X_1, \mathbb{T}) \times \dots \times \mathscr{X}_{\mathrm{det}} (X_n, \mathbb{T})$ such that each $\bm\xi_i^*$ minimizes the mean-variance functional  $\mathrm{MV}_{\gamma}(\bm{\xi} | \bm{\xi}_{-i}^*) $ over $\bm\xi\in \mathscr{X}_{\mathrm{det}} (X_i, \mathbb{T})$.
 
  \item A \emph{Nash equilibrium for CARA utility maximization} is  a collection of strategies $(\bm{\xi}_1^*, \dots, \bm{\xi}_n^*) \in \mathscr{X} (X_1, \mathbb{T}) \times \dots \times \mathscr{X}  (X_n, \mathbb{T})$ such that  each $\bm\xi_i^*$ maximizes the CARA utility functional  $U_{\gamma}(\bm{\xi} | \bm{\xi}_{-i}^*) $ over $\bm\xi\in \mathscr{X}(X_i, \mathbb{T})$.

 \end{enumerate}

\end{definition}


In the preceding definition, we have assumed that all agents share the same risk aversion parameter $\gamma\ge0$. The case in which the agents have different risk aversion parameters  is a straightforward but tedious extension of the current model. It will substantially complicate the notation while not providing  significant additional insights. For this reason, we will only consider the case of identical risk aversion parameters.   Now  let $$\varphi(t):=\mathrm{Var}  (S^0_t)\qquad\text{for $t \geq 0$.}
$$
 We define for $\theta,\gamma\ge0$,
\begin{align}
\label{Gamma}
\Gamma_{ij}^{\gamma,\theta} &= G(|t_i - t_j|)+ \gamma \varphi(t_i \wedge t_j)+2\theta\delta_{ij}, \qquad i,j=0,1,\dots,N,
\end{align}
where $\delta_{ij}$ is the Kronecker delta. Then we define
\begin{equation}
\label{Gamma_tilde}
\widetilde{\Gamma}_{ij} = 
\left\{
\begin{array}{ll}
      0 & \text{if } i < j, \\
      \frac12\Gamma_{ii}^{0,0} & \text{if } i = j, \\
     \Gamma_{ij}^{0,0}& \text{if } i > j. \\
\end{array} 
\right.
\end{equation}
Note that $\Gamma^{0,0} = \widetilde{\Gamma} + \widetilde{\Gamma}^{\top}$. We further define
\begin{align}
\label{eq:v}
\bm{v} &= \frac{1}{\bm{1}^{\top}[\Gamma^{\gamma,\theta} + (n-1) \widetilde{\Gamma}]^{-1}\bm{1}} [\Gamma^{\gamma,\theta} + (n-1) \widetilde{\Gamma}]^{-1}\bm{1},
\\
\label{eq:w}
\bm{w} &= \frac{1}{\bm{1}^{\top}[\Gamma^{\gamma,\theta} - \widetilde{\Gamma}]^{-1}\bm{1}} [\Gamma^{\gamma,\theta} - \widetilde{\Gamma}]^{-1}\bm{1}.
\end{align}
Recall that a function $g:\mathbb R\to\mathbb R$ is called strictly positive definite (in the sense of Bochner) if for all $n\in\mathbb N$ and $s_1,\dots, s_n\in\mathbb R$, the matrix $(g(s_i-s_j))_{i,j=1,\dots, n}$ is positive definite.  

\begin{assumption}\label{pos def assu}We henceforth assume that the function $\mathbb R\ni x\mapsto G(|x|)$ is strictly positive definite. 
\end{assumption}

According to P\'olya \cite{Polya49}, Assumption~\ref{pos def assu} is satisfied as soon as $G$ is convex, nonincreasing, and nonconstant (see also Young \cite{Young} for an earlier argument). It implies that the matrix $\Gamma^{0,0}$ is positive definite for all time grids $\mathbb T$. 
 As observed in~\cite{ASS}, Assumption~\ref{pos def assu}  also rules out the existence of price manipulation strategies in the sense of Huberman and Stanzl~\cite{HubermanStanzl}. 
Now we can state our first result on the existence and uniqueness of a Nash equilibrium. It extends Theorem 2.5 from~\cite{SchiedZhangHotPotato}, where the case $n=2$ and $\gamma=0$ was treated. 


\begin{theorem}\label{maintheorem} Suppose Assumption~\ref{pos def assu} holds. Then, for any time grid $\mathbb{T}$, parameters $\theta,\gamma\ge0$, initial inventories $X_1, \dots, X_n \in \mathbb{R}$, and $\bar{X} = \frac{1}{n}\sum_{j=1}^n X_j$, the strategies 
\begin{equation}
\label{opt_strat}
\bm{\xi^*_i} = \bar{X} \bm{v}+ (X_i - \bar{X}) \bm{w},\qquad i=1,\dots, n.
\end{equation}
form the unique Nash equilibrium for  mean-variance optimization. If, moreover, $S^0$ is a Bachelier model of the form \eqref{Bachelier model},  then the strategies \eqref{opt_strat}  also form a Nash equilibrium for CARA utility maximization. 
\end{theorem}

\begin{remark}
Note that the Nash equilibrium for mean-variance optimization is unique, but that we do not know whether the Nash equilibrium for CARA utility maximization is also unique. This has to do with the larger class of adapted strategies that is admitted  for CARA utility maximization. However, it follows easily from the first part of Theorem~\ref{maintheorem} that the strategies \eqref{opt_strat} form a unique Nash equilibrium for CARA utility maximization when $S^0$ is a Bachelier model and all agents are restricted to use deterministic strategies.
\end{remark}

It follows from Theorem~\ref{maintheorem} that in the following two special cases the Nash equilibrium has a particularly simple structure:
\begin{itemize}
\item if $X_1 = \dots = X_n$, then  $\bm{\xi}_i^* = X_1\bm{v}$ for $i=1,\dots, n$; 
\item if $X_1 + \dots + X_n = 0$, then $\bm{\xi}_i^* = X_i\bm{w}$ for $i=1,\dots, n$.
\end{itemize}

It was shown in Corollary 1 of~\cite{ASS} that, for convex and nonincreasing  $G$ and convex $\varphi$,  single-agent strategies ($n=1$) are always buy-only or sell-only. On the other hand, Sch\"oneborn~\cite{Schoeneborn} observed that for $G(t)=e^{-t}$, $n=2$, $\gamma=0$, and $\theta=0$ the equilibrium strategies oscillate between buy and sell orders. These oscillations are thus a genuine effect of the interaction between the two agents. This effect was explained in~\cite{SchiedStrehleZhang,SchiedZhangHotPotato,Schoeneborn} as a result of the need for protection against predatory trading by the competitor. These oscillations also have a similarity to the \lq\lq hot-potato" game between high-frequency traders during the flash crash of May 10, 2010 (see \cite[p.~2]{SEC}). In~\cite{SchiedZhangHotPotato}, the influence of $\theta$ on the oscillations of the equilibrium strategies was analyzed for $n=2$ and $\gamma=0$. It was found that there exists a critical level $\theta^*$ such that the equilibrium strategies show at least some oscillations for $\theta<\theta^*$, whereas all oscillations disappear for $\theta\ge\theta^*$. In~\cite{SchiedZhangHotPotato}, the critical level $\theta^*$ was identified as $\frac14G(0)$. Here, our goal is to analyze numerically the influence of the number $n$ of agents and the level $\gamma$ of their risk aversion on the value of $\theta^*$. 

\begin{assumption}\label{assumption 1} 
In the numerical analysis, we make the following assumptions.
\begin{enumerate}[{\rm(i)}]
\item We have $T=1$ and the time grid is equidistant:
$
\mathbb{T}_N :=  \big\{\frac{k}{N} | k = 0,1,\dots,N \big\}$ for $ N \in \mathbb{N}.$
\item $G$ is of the form $G(t) = e^{- t}$.
\item $S^0$ is a Bachelier model of the form \eqref{opt_strat} with $\sigma=1$.
\end{enumerate}
\end{assumption}

In Figure~\ref{fig:v_gamma_n2} we observe that increasing the risk aversion $\gamma$ does not stop the oscillations in the vector $\bm v$. On the contrary, increasing $\gamma$ actually magnifies the oscillations during the early trading periods.  Increasing the level $\theta$ of transaction costs, however, will clearly diminish the size of oscillations. For fixed $n$, $N$, and $\gamma$, we can therefore look for that level $\theta_{ v}=\theta_{ v}(n,N,\gamma)$ at which $\min_i v_i$ becomes nonnegative. Figure \ref{fig:thetav} suggests that, at least for sufficiently large $N$,  this level is completely independent of the risk aversion parameter $\gamma$, an observation we find highly surprising. 
In Figure~\ref{fig:thetavn2}, we provide numerical surface plots for the function $(N,\gamma)\mapsto\theta_v(n,N,\gamma)$ with $n=2$ and $n=5$. Together with additional simulations carried out by the authors, Figures~\ref{fig:thetavn2} and~\ref{fig:thetav} 
suggest that for each $n$ there is a critical level at which all oscillations of $\bm v$ cease and that it is given by
\begin{equation}\label{conj 1}
\theta^*_v(n)=\sup_{N,\gamma}\theta_v(n,N,\gamma)=\frac{n-1}4.
\end{equation}
This conjecture is consistent with the theoretical results obtained in~\cite{ASS} and~\cite{SchiedZhangHotPotato} for $n=1$ and $n=2$, respectively.
\begin{figure}[H]
\centering
\includegraphics[width=\textwidth]{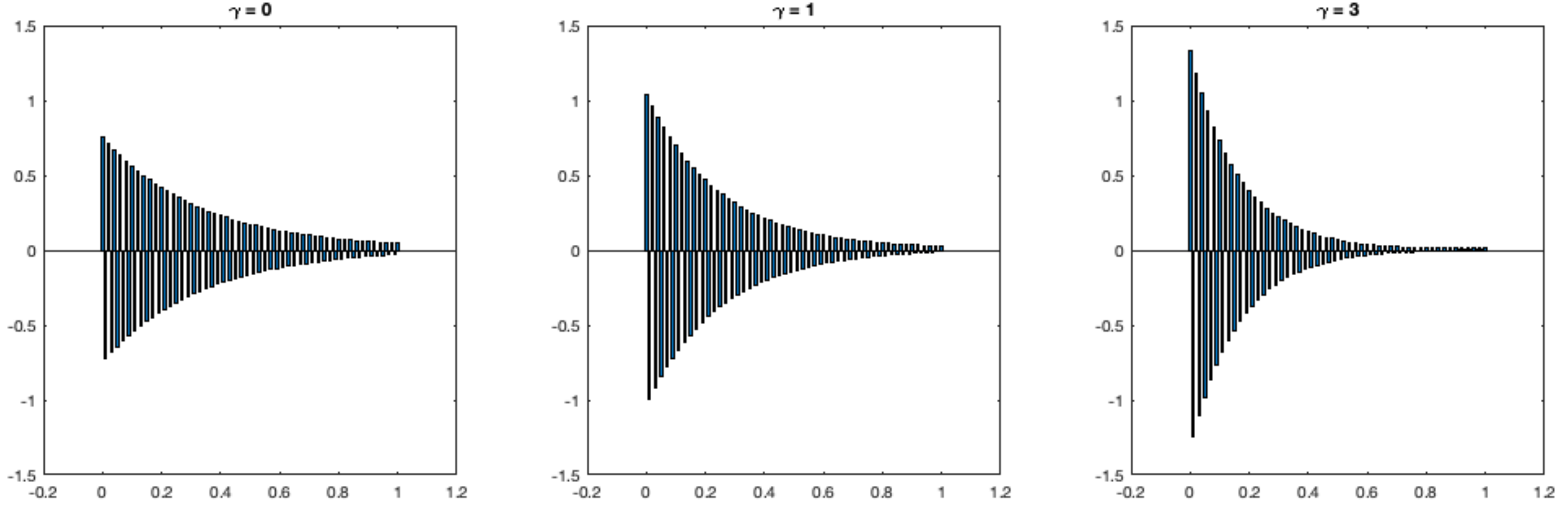}
\caption{Vector $\bm{v}$ with $\gamma = 0$ (left), $\gamma = 1$ (middle) and $\gamma = 3$ (right) for $n = 2$, $N= 100$,  and $\theta = 0$.}
\label{fig:v_gamma_n2}
\end{figure}

\begin{figure}[H]
\centering
\includegraphics[width=\textwidth]{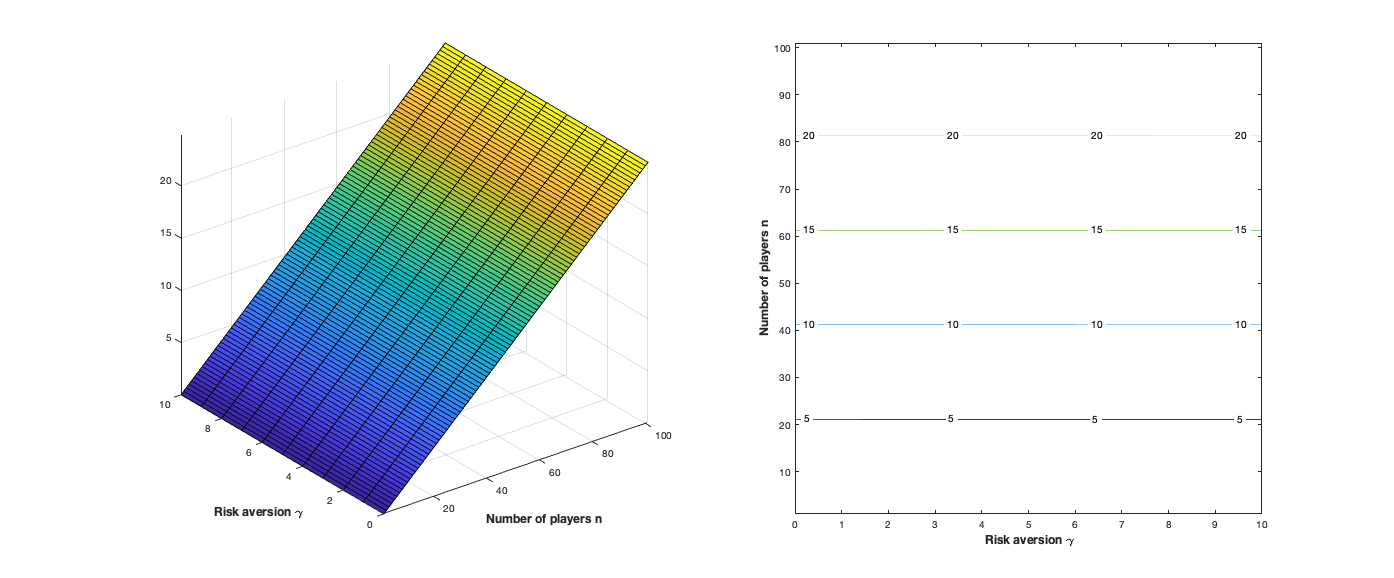}
\caption{Surface plot (left) and level curves (right) of $\theta_v(n,N,\gamma)$ with respect to the number of players $n$ and the risk aversion parameter $\gamma$ for $N= 500$.}
\label{fig:thetav}
\end{figure}

\begin{figure}[H]
\centering
\includegraphics[width=\textwidth]{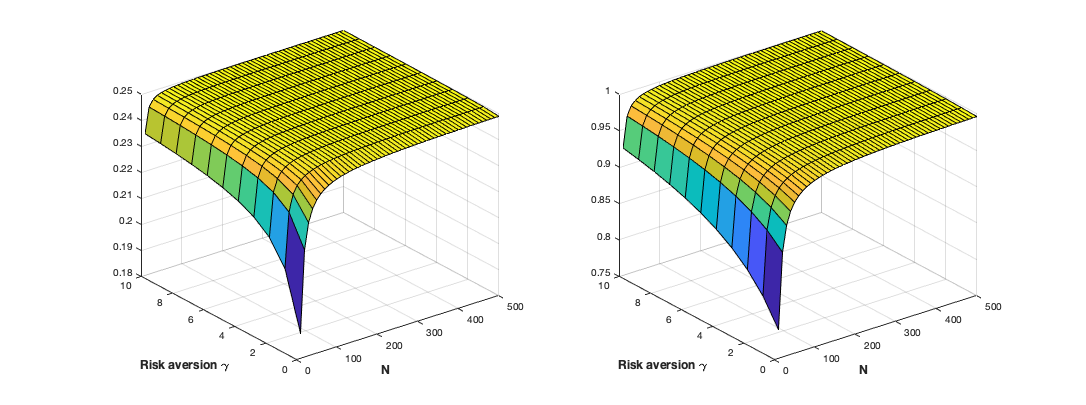}
\caption{Surface plots of $\theta_v(n,N,\gamma)$ with $n = 2$ (left) and $n=5$ (right) with respect to $N$ and the risk aversion parameter $\gamma$.}
\label{fig:thetavn2}
\end{figure}



Now we turn to vector $\bm{w}$. The first observation is that $\bm w$ is independent of the number $n$ of agents. Thus, the critical level $\theta_w$ at which oscillations cease is a function of $N$ and $\gamma$ only. For $\gamma=0$, $\bm w$ must be identical to the one studied in~\cite{SchiedStrehleZhang,SchiedZhangHotPotato}, and it follows from Theorem 2.7 of \cite{SchiedZhangHotPotato} that the critical transaction cost level in this case is $\frac14$. Moreover, it can be seen from Figures~\ref{fig:w_gamma} and~\ref{fig:w_neg_vs_gamma} that, in contrast to $\bm v$, the oscillations of the vector $\bm w$ are influenced by changing the risk aversion $\gamma$. More precisely, increasing $\gamma$ does have a diminishing effect on the oscillations of $\bm w$. Therefore, we conjecture that
\begin{equation}\label{conj 2}
\theta^*_w=\sup_{N,\gamma}\theta_w(N,\gamma)=\frac14.
\end{equation}
This conjecture is also supported by Figure~\ref{fig:thetaw} and consistent with Theorem 2.7 of \cite{SchiedZhangHotPotato}.

\begin{remark}
For simplicity, we performed the above simulations under the assumption that $G$ is of the form $G(t) = e^{- t}$. One can instead assume that $G(t) = e^{- \rho t}$ for some $\rho > 0$ and verify that the above numerical results remain valid. Note that for large $\rho$, one will need a sufficiently large $N$ to visualize the convergence of $\theta$ to the critical values. Moreover, we have extended the simulations to a power law decay kernel of the form $G(t) = 1/(1+t)^p$ with $p > 0$. The behaviours of vector $\bm{v}$ and $\bm{w}$ as well as the critical values for $\theta$ are again consistent with the above analysis. The corresponding plots are therefore omitted.
\end{remark}

\begin{remark}As mentioned above, conjectures \eqref{conj 1} and \eqref{conj 2} are proved in \cite{SchiedZhangHotPotato} for the special case of two risk-neutral agents. Already in this special case, the proof of \eqref{conj 1} is quite involved. It relies in part on the fact that $\Gamma^{0,0}$ is a Kac--Murdock--Szeg\H o matrix, whose inverse is known explicitly.  For $\gamma>0$, however, the inverse of $\Gamma^{\gamma,\theta}$ is not known, and so the proof method from \cite{SchiedZhangHotPotato}. The proof of \eqref{conj 2} in the special case of  \cite{SchiedZhangHotPotato} relies on the fact that $\Gamma^{0,\theta}-\widetilde\Gamma$ is an upper triangular Toeplitz matrix. The inverse of such a matrix can be computed by the coefficients of the reciprocal of the power series formed from the coefficients of  $\Gamma^{0,\theta}-\widetilde\Gamma$. Then the celebrated Kalusza sign criterion is applied in \cite{SchiedZhangHotPotato} so as to characterize the case in which all coefficients of the reciprocal power series are nonnegative. For $\gamma>0$, however, $\Gamma^{0,\theta}-\widetilde\Gamma$ is no longer an upper triangular Toeplitz matrix, and so the proof from \cite{SchiedZhangHotPotato} cannot be extended to our present situation.
\end{remark}

\begin{figure}[H]
\centering
\includegraphics[width=\textwidth]{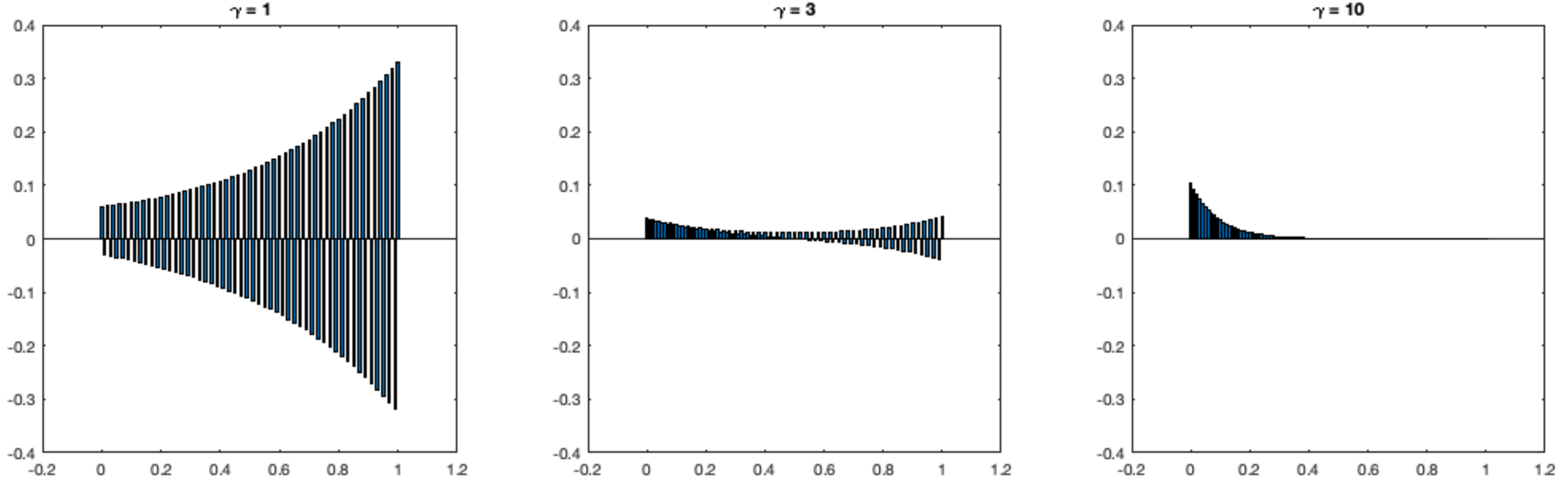}
\caption{Vector $\bm{w}$ with $\gamma = 1$ (left), $\gamma = 3$ (middle) and $\gamma = 10$ (right) for $N= 100$,  and $\theta = 0$.}
\label{fig:w_gamma}
\end{figure}

\begin{figure}[H]
\centering
\includegraphics[scale=0.25]{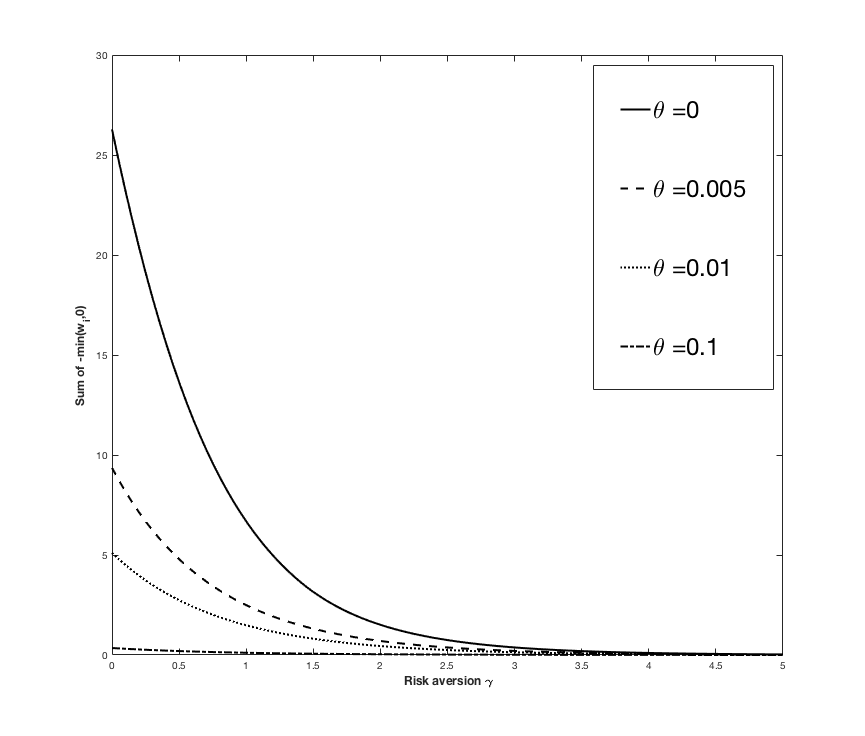}
\caption{Magnitude of the sum of negative components in vector $\bm{w}$ as a function of $\gamma$ with $\theta \in \{0,0.005,0.01,0.1\}$ for $N= 100$.}
\label{fig:w_neg_vs_gamma}
\end{figure}

\begin{figure}[H]
\centering
\includegraphics[scale=0.23]{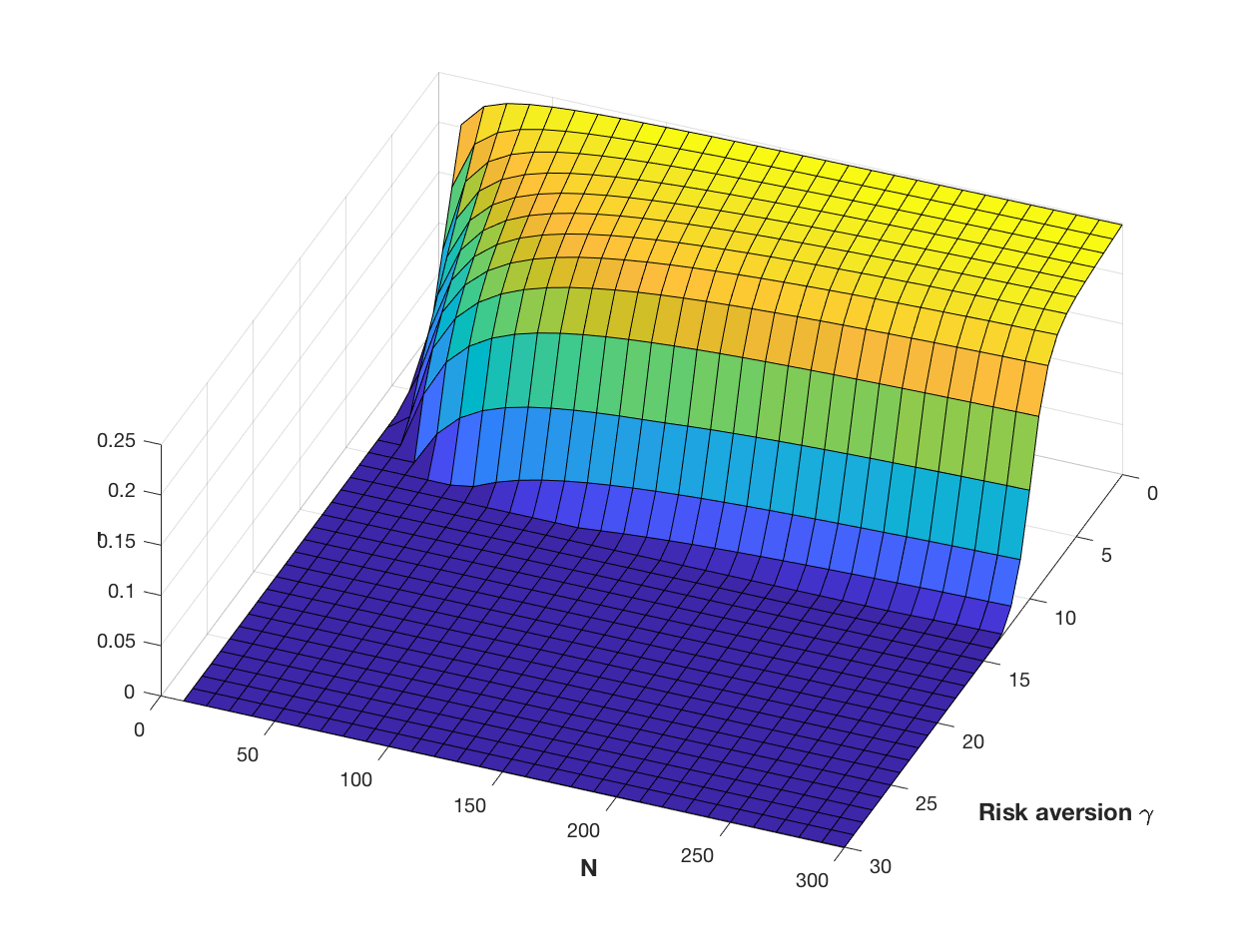}
\caption{Surface plot of $\theta_w(N,\gamma)$ with respect to $N$ and the risk aversion parameter $\gamma$.}
\label{fig:thetaw}
\end{figure}

It is interesting to note that for $\gamma=0$ and $\theta=\frac14$, the  vector $\bm w$ has a particularly simple structure. This is stated in the following theoretical result.

\begin{proposition}
\label{prop:w_rho}
Under Assumption~\ref{assumption 1} (i), (iii), and $G(t) = e^{- \rho t}$ with $\rho > 0$, for $\theta =\frac{1}{4}$ and $\gamma = 0$, 
\[
w_0=\cdots=w_{N-1}= \frac{1-e^{-\rho/N}}{N(1-e^{-\rho/N})+1}\qquad\text{and}\qquad w_N= \frac{1}{N(1-e^{-\rho/N})+1}.
\]
\end{proposition}


\subsection{Infinite time horizon}
\label{sec:inf}

For non-vanishing risk aversion $\gamma>0$, it is possible to study our problem also for an infinite time horizon. The intuitive reason is that any risk-averse investor will automatically try to liquidate any position held in an asset whose price process is a martingale.

\begin{assumption}\label{assumption 2}Throughout Section~\ref{sec:inf}, we make the following assumptions.
\begin{enumerate}[{\rm(i)}]
\item The time grid is $\mathbb N_0=\{0,1,\dots\}$.
\item $G$ is of the form $G(t) = e^{- \rho t}$ for some $\rho>0$.
\item $S^0$ is a Bachelier model of the form \eqref{opt_strat}.
\end{enumerate}
\end{assumption}

Under Assumption~\ref{assumption 2} (i), the strategy of an agent $i$ with initial position $X_i\in\mathbb R$ will be represented by a sequence $\bm\xi=(\xi_0,\xi_1,\dots)$ of random variables such that the following conditions are satisfied:
\begin{itemize}
\item each $\xi_i$ is $\mathscr F_i$-measurable;
\item the random variable $\bm\xi$ takes values in the space $\ell^1$ of absolutely summable real sequences;
\item the random variable $\bm\xi$ is  bounded in the Banach space $\ell^\infty$ of bounded sequences;
\item we have $\sum_{k=0}^\infty\xi_k(\omega)=X_i$ for each $\omega\in\Omega$.
\end{itemize}
   The set of all these strategies will be denoted by $\mathscr X(X_i,\mathbb N_0)$. Again, the class of all deterministic strategies in $\mathscr X(X_i,\mathbb N_0)$ will be denoted by $\mathscr X_{\mathrm{det}}(X_i,\mathbb N_0)$. Since $\ell^1\subset\ell^2$, it is clear that \eqref{cost functional} can be extended as follows to strategies $\bm\xi_i\in\mathscr X(X_i,\mathbb N_0)$, $i=1,\dots,n$,
 $$ \mathscr{C}_{\mathbb{N}_0}(\bm{\xi}_i  | \bm{\xi}_{-i} )= 
 \sum^\infty_{k=0}\Big[\frac{G(0)}{2}\xi_{i,k}^2 - S^{\Xi}_{t_k}\xi_{i,k}  + \frac{G(0)}2\xi_{i,k} \sum_{j\neq i}\xi_{j,k}  + \theta\xi_{i,k}^2\Big]. $$
Again, each agent will aim to minimize the following mean-variance functional,
$$\mathrm{MV}_{\gamma}(\bm{\xi}_i  | \bm{\xi}_{-i}) = \mathbb{E}[\mathscr{C}_{\mathbb{N}_0}(\bm{\xi}_i  | \bm{\xi}_{-i})] + \frac{\gamma}{2}\mathrm{Var}  [\mathscr{C}_{\mathbb{N}_0}(\bm{\xi}_i  | \bm{\xi}_{-i})],\qquad \bm\xi_i\in\mathscr X_{\mathrm{det}}(X_i,\mathbb N_0),
$$
or to maximize the CARA utility functional
$$U_\gamma(\bm\xi_i|\bm\xi_{-i})=\mathbb E[u_\gamma(-\mathscr{C}_{\mathbb{N}_0}(\bm{\xi}_i  | \bm{\xi}_{-i}))],\qquad \bm\xi_i\in\mathscr X(X_i,\mathbb N_0).
$$
The notion of  Nash equilibria for mean-variance optimization and CARA utility maximization can be defined in exactly  the same way as in Definition~\ref{def:nash}.  However, it is not clear \emph{a priori} whether the formulas \eqref{eq:v} and \eqref{eq:w} for $\bm v$ and $\bm w$ can also be extended to an infinite time horizon, because it is no longer clear whether the vector $\bm 1$ belongs to the range of the linear operators $\Gamma^{\gamma,\theta} +(n-1)\widetilde{\Gamma}$ and $\Gamma^{\gamma,\theta} -\widetilde{\Gamma}$. 
The following result states the existence of an infinite-horizon Nash equilibrium in a specific situation.



\begin{theorem}
\label{thm:infv}In addition to Assumption~\ref{assumption 2}, suppose that $\gamma>0$  and 
\begin{equation}\label{critical theta}
\theta = \theta^* := \frac{n-1}{4}.
\end{equation}
Then there exist  unique positive solution $\alpha$ and $\beta$ of the two equations
\begin{align}
\label{eq:infv}
0&=\frac{1}{e^{(\alpha + \rho)}-1}-\frac{n}{e^{(\alpha - \rho)}-1} - \frac{\gamma\sigma^2 e^{-\alpha}}{(1-e^{-\alpha})^2} ,\\
0&=2\theta + \frac{1}{2} + \frac{1}{e^{(\beta+\rho)}-1}-\frac{\gamma\sigma^2 e^{-\beta}}{(1-e^{-\beta})^2}.\label{eq:infw}
\end{align}
Moreover, $\alpha\in(0,\rho)$. For these, we define $\bm v\in\ell^1$ through
\begin{align*}
v_0=\frac{e^\alpha-1}{e^\alpha-e^{\alpha-\rho}}\qquad\text{and, for $i=1,2,\dots$,}\qquad v_i=\frac{e^{-\alpha i}}{\frac1{e^\alpha-1}+\frac1{1-e^{\alpha-\rho}}}
\end{align*}
and $\bm w\in\ell^1$ through
$$w_i=(1-e^{-\beta})e^{-\beta i}.
$$
Then, for any initial positions $X_1,\dots, X_n$, the strategies 
\begin{equation}
\label{opt_strat_infinity}
\bm{\xi^*_i} = \bar{X} \bm{v}+ (X_i - \bar{X}) \bm{w},\qquad i=1,\dots, n.
\end{equation}
form a Nash equilibrium for  mean-variance optimization in $\mathscr X_{\mathrm{det}}(X_1,\mathbb N_0)\times\cdots\times \mathscr X_{\mathrm{det}}(X_n,\mathbb N_0)$ and a Nash equilibrium for CARA utility maximization in  $\mathscr X (X_1,\mathbb N_0)\times\cdots\times \mathscr X(X_n,\mathbb N_0)$.
\end{theorem}

In the preceding theorem, we have assumed that $\theta=\theta^*$. In the general case, the following result will follow immediately from the proof of Theorem \ref{thm:infv}.

\goodbreak 

\begin{proposition}\label{nonexist prop}In addition to Assumption~\ref{assumption 2}, suppose that $\gamma>0$.
\begin{enumerate}
\item If $X_1+\cdots+X_n=0$, then  the formula $\bm\xi^*_i=X_i\bm w$, for $\bm w$  as in Theorem \ref{thm:infv},
 provides a Nash equilibrium for  every choice  of $\theta\ge0$.
 \item If $X_1=\cdots=X_n$, then there is no Nash equilibrium whose strategies decay exponentially in time unless the condition $\theta = \theta^* := \frac{n-1}{4}$ is satisfied.
\end{enumerate}
\end{proposition}

As a matter of fact, we conjecture that in the situation of Proposition \ref{nonexist prop} (b), no Nash equilibrium exists unless $\theta=\theta^*$ holds. The situation is very similar to the one of Theorem 4.5 in~\cite{SchiedStrehleZhang}, where a continuous-time version of the game for $n=2$ and $\gamma=0$ was analyzed. It was shown  there that a continuous-time Nash equilibrium can exist only if $\theta=\theta^*$ or $X_1=X_2=0$. In both situations, the underlying intuition for the nonexistence of Nash equilibria results from the possibility of trading infinitely often, either in continuous time or over an infinite time horizon. This shows that the idealization of admitting infinitely many trades is not as harmless as it might seem, an observation that has also been made, for instance, in the context of the FTAP.

The qualitative behavior of the respective solutions $\alpha$ and $\beta$ of  \eqref{eq:infv} and  \eqref{eq:infw} is plotted in Figures~\ref{fig:alpha_rho_n} and~\ref{fig:beta_theta_inf}. In case $n=1$, we have the following explicit result. 
\begin{proposition}
\label{prop:infv}
If $n = 1$, then the solution $\alpha$ of \eqref{eq:infv} is given by
\begin{equation}
\label{eq:infv_n1}
\alpha = \cosh^{-1}\Big[\frac{\gamma \sigma^2 \cosh(\rho)+2\sinh(\rho)}{\gamma  \sigma^2 +2\sinh(\rho)}\Big].
\end{equation}
\end{proposition}

  \begin{figure}[H]
  \begin{minipage}{9cm}
\includegraphics[width=9cm]{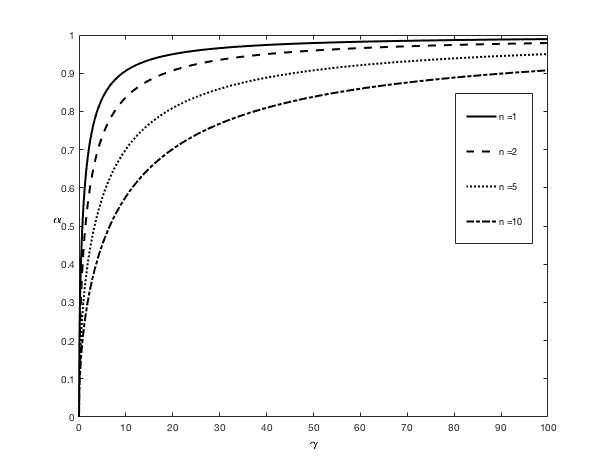}
\end{minipage}\quad
 \begin{minipage}{9cm}
\includegraphics[width=9cm]{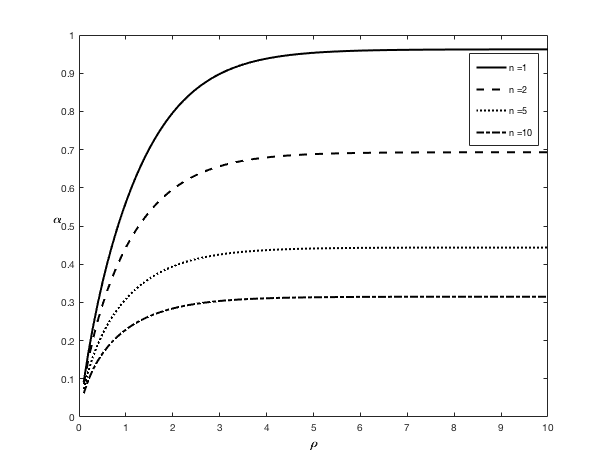}
\end{minipage}

\caption{The solution $\alpha$ of \eqref{eq:infv} as a function of $\gamma$ (left) and $\rho$ (right) with different values of $n$,  $\theta = \frac{n-1}{4}$, and the respective remaining parameters set to 1. }
\label{fig:alpha_rho_n}
\end{figure}

 \begin{figure}[H]
 \begin{minipage}{6cm}
\includegraphics[width=6cm,height=5cm]{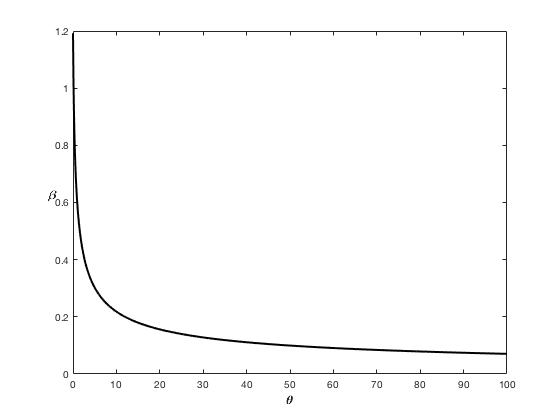}
\end{minipage}\quad
 \begin{minipage}{6cm}
\includegraphics[width=6cm,height=5cm]{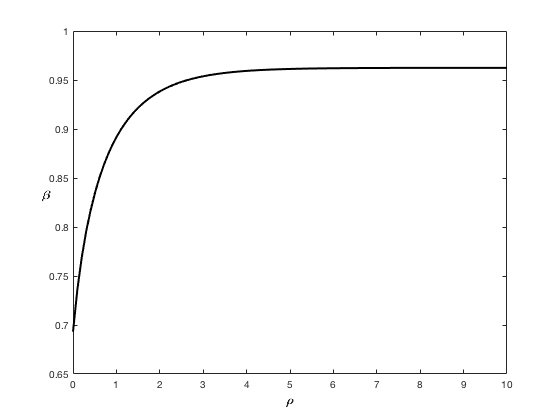}
\end{minipage}\quad
 \begin{minipage}{6cm}
\includegraphics[width=6cm,height=5cm]{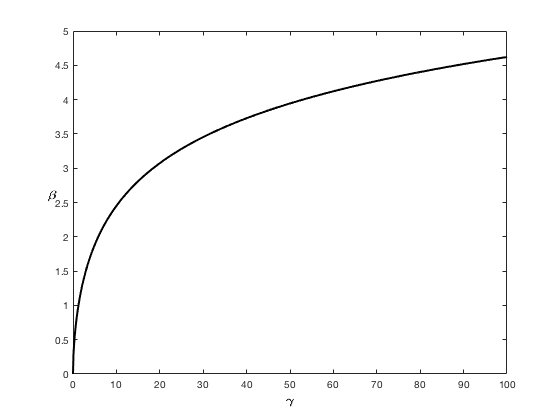}
\end{minipage}
\caption{The solution $\beta$ of \eqref{eq:infw} as a function of $\theta$ (left), $\rho$ (center), and $\gamma$ (right), with respective remaining parameters set to $1$.}

\label{fig:beta_theta_inf}

\end{figure}


\section{Proofs}\label{proofs section}


\begin{lemma}
\label{lemma:mvf}
 An admissible strategy $\bm{\xi}_i  \in \mathscr{X}_{\mathrm{\rm det}}(X_i,\mathbb{T})$ given all the competitors' strategies $\bm{\xi}_j  \in \mathscr{X}_{\mathrm{\rm det}}(X_j,\mathbb{T})$ with $j \neq i$ has the following mean-variance functional:
\begin{equation}
\mathrm{MV}_{\gamma}(\bm{\xi}_i  | \bm{\xi}_{-i}) = - X_i S^0_0 + \frac{1}{2}\bm{\xi}_i ^{\top}\Gamma^{\gamma,\theta}\bm{\xi}_i +\bm{\xi}_i ^{\top}\widetilde{\Gamma} \sum_{j\neq i}\bm{\xi}_j .
\end{equation}
\end{lemma}

\begin{proof}Since all strategies are deterministic,
\[
\begin{split}
\mathbb{E}[\mathscr{C}_{\mathbb{T}}(\bm{\xi}_i  | \bm{\xi}_{-i})] 
& = \sum^N_{k=0}\Big[\frac{G(0)}{2}\xi_{i,k}^2 - \mathbb{E} [S^{\Xi}_{t_k}\xi_{i,k}] + \frac{G(0)}2\sum_{j\neq i}\xi_{i,k}\xi_{j,k}  + \theta\xi_{i,k}^2\Big].
\end{split}
\]

Since $S^0$ is a martingale,
\[
\begin{split}
\sum^N_{k=0} \mathbb{E} [S^{\Xi}_{t_k}\xi_{i,k}] 
&  = \sum^N_{k=0}\xi_{i,k}\mathbb{E}[S^0_{t_k}] - \sum^N_{k=0}\xi_{i,k}\sum_{m=0}^{k-1} \Big(G(t_k-t_m) \sum_{j=1}^n \xi_{j,m} \Big) \\
& = X_iS^0_0 - \sum^N_{k=0}\xi_{i,k}\sum_{m=0}^{k-1} \xi_{i,m}G(t_k-t_m) \\
& \qquad - \sum^N_{k=0}\xi_{i,k}\sum_{m=0}^{k-1} G(t_k-t_m)\sum_{j\neq i} \xi_{j,m} .
\end{split}
\]
Moreover, using matrix notation,
\[
\frac{G(0) + 2\theta}{2}\sum^N_{k=0}\xi_{i,k}^2 + \sum^N_{k=0}\xi_{i,k}\sum_{m=0}^{k-1} \xi_{i,m}G(t_k-t_m)
 = \frac{1}{2} \Big[ 2\theta \sum^N_{k=0}\xi_{i,k}^2 + \sum^N_{k,m=0}\xi_{i,k}\xi_{i,m}G(|t_k-t_m|) \Big]  = \frac{1}{2}\bm{\xi}_i ^{\top}\Gamma^{0,\theta}\bm{\xi}_i ,
\]
and
\[
\sum^N_{k=0}\xi_{i,k} \Big[ \frac{G(0)}2\sum_{j\neq i}\xi_{j,k}+ \sum_{m=0}^{k-1} G(t_k-t_m)\sum_{j\neq i} \xi_{j,m}  \xi_{j,m}\Big] = \bm{\xi}_i ^{\top}\widetilde{\Gamma}\sum_{j\neq i}\bm{\xi}_j .
\]
Using again that $\bm{\xi}_i $ are deterministic and the martingale property of $S^0$,
\[
\mathrm{Var}  [\mathscr{C}_{\mathbb{T}}(\bm{\xi}_i  | \bm{\xi}_{-i})] = \mathrm{Var}   \Big[\sum^N_{k=0} S^{\Xi}_{t_k}\xi_{i,k} \Big] = \mathrm{Var}   \Big[\sum^N_{k=0} S^0_{t_k}\xi_{i,k} \Big]= \sum^N_{p,q=0}\xi_p\xi_q\text{Cov}(S^0_{t_p},S^0_{t_q}) = \sum^N_{p,q=0}\xi_p\xi_q\varphi(t_p \wedge t_q).
\]
By substituting the preceding results into \eqref{eq:mvf}, we obtain the desired formula:
\[
\begin{split}
\mathrm{MV}_{\gamma}(\bm{\xi}_i  | \bm{\xi}_{-i}) 
& = \mathbb{E}[\mathscr{C}_{\mathbb{T}}(\bm{\xi}_i  | \bm{\xi}_{-i})] + \frac{\gamma}{2}\mathrm{Var}  [\mathscr{C}_{\mathbb{T}}(\bm{\xi}_i  | \bm{\xi}_{-i})] \\
& =  - X_i S^0_0 + \frac{1}{2}\bm{\xi}_i ^{\top}\Gamma^{0,\theta}\bm{\xi}_i + \bm{\xi}_i ^{\top}\widetilde{\Gamma}\sum_{j\neq i}\bm{\xi}_j +\frac{\gamma}{2} \sum^N_{p,q=0}\xi_p\xi_q\varphi(t_p \wedge t_q) \\
& =  - X_i S^0_0 + \frac{1}{2}\bm{\xi}_i ^{\top}\Gamma^{\gamma,\theta}\bm{\xi}_i +\bm{\xi}_i ^{\top}\widetilde{\Gamma} \sum_{j\neq i}\bm{\xi}_j .
\end{split}
\]
\end{proof}


We will use the convention of saying that an $n\times n$-matrix $A$ is \emph{positive}  if $\bm x^\top A\bm x>0$ for all nonzero $\bm x\in\mathbb{R}^{n}$, which makes sense also if $A$ is not necessarily symmetric. Clearly, for a positive   matrix $A$ there is no  nonzero $\bm x\in\mathbb{R}^{n}$ for which $A\bm x=\bm0$, and so $A$ is invertible.

\begin{lemma}
\label{lemma:invertible}
For all $\gamma, \theta \geq 0$ and $n \geq 1$, the matrices  $\Gamma^{\gamma,\theta}$, $\widetilde{\Gamma}$, $\Gamma^{\gamma,\theta} - \widetilde{\Gamma}$ and $\Gamma^{\gamma,\theta} + (n-1) \widetilde{\Gamma}$ are positive.
\begin{proof}
By Lemma 3.2 in~\cite{SchiedZhangHotPotato}, the matrices $\Gamma$, $\Gamma^{0,\theta}$, $\widetilde{\Gamma}$, and $\Gamma^{0,\theta} - \widetilde{\Gamma}$ are positive. Since $C:=(\varphi(t_i \wedge t_j))_{i,j=1,\dots, N}$ is the covariance matrix of the random variables, $S_{t_1}^0,\dots.S^0_{t_N}$, it is nonnegative definite. It follows that $\Gamma^{\gamma,\theta}=\Gamma^{0,\theta}+\gamma C$ and $\Gamma^{\gamma,\theta} - \widetilde{\Gamma} = (\Gamma^{0,\theta} - \widetilde{\Gamma}) + \gamma C$ are  positive as well. Hence, $\Gamma^{\gamma,\theta} + (n-1) \widetilde{\Gamma}$ is also positive. 
\end{proof}
\end{lemma}


\begin{lemma}
\label{lemma:uniqueNash}
For a given time grid $\mathbb{T}$ and initial values $X_1, \dots, X_n \in \mathbb{R}$, there exists at most one Nash equilibrium for mean-variance optimization.
\end{lemma}

\begin{proof}The proof is similar to the one  of Lemma 3.3 in~\cite{SchiedZhangHotPotato} or Lemma 4 in~\cite{SchiedZhangCARA} and hence omitted.
\end{proof}


Now we are ready to prove Theorem~\ref{maintheorem}.

\begin{proof}[Proof of Theorem~\ref{maintheorem}] By  Definition~\ref{def:nash} and Lemma~\ref{lemma:mvf}, we have the following  linear-quadratic optimization problem: for all $i = 1, \dots, n$,
\[
\mathrm{MV}_{\gamma}(\bm{\xi}_i^* | \bm{\xi}_{-i}^*) = - X_i S^0_0+\min_{\bm{\xi}_i  \in \mathscr{X}_{\mathrm{det}}(X_i, \mathbb{T})} \Big(   \frac{1}{2}\bm{\xi}_i ^{\top}\Gamma^{\gamma,\theta}\bm{\xi}_i  + \bm{\xi}_i ^{\top}\widetilde{\Gamma} \sum_{j\neq i}\bm{\xi}_j\Big) .
\]
The constraint $\bm{\xi}_i  \in \mathscr{X}_{\mathrm{det}}(X_i, \mathbb{T})$ can be re-written as the linear equality constraint $
\bm{1}^{\top}\bm{\xi}_i  = X_i$.

To solve this problem, we use the Lagrange multiplier theorem  \cite[pp.~276-283]{Bertsekas} to obtain $\alpha_i \in \mathbb{R}$ for $i = 1, \dots, n$ such that the optimal strategies satisfy the following necessary conditions:
\begin{equation}
\label{pf:optimalityconditions}
\left\{
\begin{array}{ll}
      \Gamma^{\gamma,\theta}\bm{\xi}_i^* + \widetilde{\Gamma} \sum_{j\neq i}\bm{\xi}_j^* = \alpha_i \bm{1}, \\
      \\
      \bm{1}^{\top}\bm{\xi}_i^* = X_i. 
\end{array} 
\right.
\end{equation}
We will show below that these equations are also sufficient for our optimization problem. 
Summing over $i$ in the first line of \eqref{pf:optimalityconditions} yields
\[
[\Gamma^{\gamma,\theta} + (n-1) \widetilde{\Gamma}]\sum^{n}_{j=1}\bm{\xi}_j^* = \sum^{n}_{j=1}\alpha_j \bm{1}.
\]
By Lemma~\ref{lemma:invertible}, $\Gamma^{\gamma,\theta}+ (n-1) \widetilde{\Gamma}$ is an invertible matrix. Thus,
\begin{equation}
\label{pf:maineq1}
\begin{split}
\sum^{n}_{j=1}\bm{\xi}_j^* 
& = \sum^{n}_{j=1}\alpha_j [\Gamma^{\gamma,\theta}+ (n-1) \widetilde{\Gamma}]^{-1}\bm{1} \\
& = \frac{\bm{1}^{\top} \sum^{n}_{j=1}\alpha_j [\Gamma^{\gamma,\theta}+ (n-1) \widetilde{\Gamma}]^{-1}\bm{1}}{\bm{1}^{\top}[\Gamma^{\gamma,\theta}+ (n-1) \widetilde{\Gamma}]^{-1}\bm{1}} [\Gamma^{\gamma,\theta}+ (n-1) \widetilde{\Gamma}]^{-1}\bm{1} \\
& =  \sum^{n}_{j=1}\frac{\bm{1}^{\top}\bm{\xi}_j^*}{\bm{1}^{\top}[\Gamma^{\gamma,\theta}+ (n-1) \widetilde{\Gamma}]^{-1}\bm{1}} [\Gamma^{\gamma,\theta}+ (n-1) \widetilde{\Gamma}]^{-1}\bm{1}\\
&= \sum^{n}_{j=1} X_j \bm{v},
\end{split}
\end{equation}
where we have used  the second condition from \eqref{pf:optimalityconditions} in the final step. 

Now consider the first conditions in \eqref{pf:optimalityconditions}. Pick the $i^{th}$ equation, multiply by $(n-1)$, and then subtract the other $(n-1)$ equations from it. We get
\[
\Gamma^{\gamma,\theta} \Big[ (n-1)\bm{\xi}_i^* - \sum_{j\neq i}\bm{\xi}_j^* \Big]  - \widetilde{\Gamma}  \Big[ (n-1)\bm{\xi}_i^* - \sum_{j\neq i}\bm{\xi}_j^* \Big]  = \Big[ (n-1)\alpha_i - \sum_{j\neq i}\alpha_j \Big] \bm{1}.
\]
Further simplifications show that
\[
(\Gamma^{\gamma,\theta} - \widetilde{\Gamma})\Big[ n\bm{\xi}_i^* - \sum^{n}_{j=1}\bm{\xi}_j^* \Big] = \Big[ n\alpha_i - \sum^{n}_{j=1}\alpha_j \Big] \bm{1}.
\]
The matrix $\Gamma^{\gamma,\theta} - \widetilde{\Gamma}$ is invertible by Lemma~\ref{lemma:invertible}. If follows that
\begin{equation}
\label{pf:maineq2}
n\bm{\xi}_i^* - \sum^{n}_{j=1}\bm{\xi}_j^* = \Big[ nX_i - \sum^{n}_{j=1}X_j \Big] \bm{w}.
\end{equation}
Using \eqref{pf:maineq1} and \eqref{pf:maineq2} now gives
\[
\bm{\xi}_i^* = \bar{X} \bm{v}+ (X_i - \bar{X}) \bm{w}
\]
where $\bar{X} = \frac{1}{n}\sum_{j=1}^n X_j$.

Now we show that the equations \eqref{pf:optimalityconditions} are sufficient for the minimization of our mean-variance functional. To this end, we rewrite the objective mean-variance functional as follows:
\[
\frac{1}{2}\bm{\xi}_i ^{\top}\Gamma^{\gamma,\theta}\bm{\xi}_i  + \bm{\xi}_i ^{\top}\widetilde{\Gamma} \sum_{j\neq i}\bm{\xi}_j^* = \frac{1}{2}\bm{\xi}_i ^{\top}\Gamma^{\gamma,\theta}\bm{\xi}_i  + \bm{g_i}^{\top}\bm{\xi}_i ,
\]
where $\bm{g_i} = \widetilde{\Gamma}\sum_{j\neq i}\bm{\xi}_j^*$.
Next, for $i = 1, \dots, n$, we consider  arbitrary $\bm{\eta_i} \in \mathscr X_{\mathrm{det}}(X_i,\mathbb T)$. Then, by \eqref{pf:optimalityconditions},
\begin{equation}\label{optimality check eq}
\begin{split}
\frac{1}{2}\bm{\eta_i}^{\top}\Gamma^{\gamma,\theta}\bm{\eta_i} + \bm{g_i}^{\top}\bm{\eta_i} - \Big[ \frac{1}{2}{\bm{\xi}_i^*}^{\top}\Gamma^{\gamma,\theta}{\bm{\xi}_i^*} + \bm{g_i}^{\top}{\bm{\xi}_i^*} \Big]
& = \frac{1}{2}(\bm{\eta_i} + {\bm{\xi}_i^*})^{\top}\Gamma^{\gamma,\theta}(\bm{\eta_i} - {\bm{\xi}_i^*}) + \bm{g_i}^{\top}(\bm{\eta_i} - {\bm{\xi}_i^*}) \\
& = \Big[\frac{1}{2} (\Gamma^{\gamma,\theta})^{\top}(\bm{\eta_i} + {\bm{\xi}_i^*}) +  \bm{g_i}\Big]^{\top}(\bm{\eta_i} - {\bm{\xi}_i^*}) \\
& = \Big[(\Gamma^{\gamma,\theta}{\bm{\xi}_i^*} +  \bm{g_i}) + \frac{1}{2}(\Gamma^{\gamma,\theta})^\top(\bm{\eta_i} - {\bm{\xi}_i^*}) \Big]^{\top}(\bm{\eta_i} - {\bm{\xi}_i^*}) \\
& = \Big[\alpha_i \bm{1} + \frac{1}{2}(\Gamma^{\gamma,\theta})^\top(\bm{\eta_i} - {\bm{\xi}_i^*}) \Big]^{\top}(\bm{\eta_i} - {\bm{\xi}_i^*}) \\
& = \alpha_i \bm{1}^{\top}(\bm{\eta_i} - {\bm{\xi}_i^*}) + \frac{1}{2}(\bm{\eta_i} - {\bm{\xi}_i^*}) ^{\top}\Gamma^{\gamma,\theta}(\bm{\eta_i} - {\bm{\xi}_i^*}) \\
& \geq 0,
\end{split}
\end{equation}
with equality if and only if $\bm{\eta_i} = {\bm{\xi}_i^*}$. Altogether, we obtain that  \eqref{opt_strat} defines the unique Nash equilibrium in $\mathscr{X}_{\text{det}}(X_1, \mathbb{T}) \times \dots \times \mathscr{X}_{\text{det}}(X_n, \mathbb{T})$.

Now we turn to CARA utility maximization.  We first note that the cost functional $\mathscr C_{\mathbb T}(\bm\xi_i|\bm\xi_{-i})$ is a Gaussian random variable if $S^0$ is a Bachelier model and the strategies $\bm\xi_1,\dots\bm\xi_n$ are deterministic. Therefore, for $\gamma>0$,
$$U_\gamma(\bm\xi_i|\bm\xi_{-i})=\frac1\gamma\Big(1-\exp\Big(\gamma\mathbb E[\mathscr C_{\mathbb T}(\bm\xi_i|\bm\xi_{-i})]+\frac{\gamma^2}2\text{Var}[\mathscr C_{\mathbb T}(\bm\xi_i|\bm\xi_{-i})]\Big)\Big)=u_\gamma\big(-\mathrm{MV}_\gamma(\bm\xi_i|\bm\xi_{-i})\big).
$$
For $\gamma=0$, we clearly have 
$$U_0(\bm\xi_i|\bm\xi_{-i})=-\mathbb E[\mathscr C_{\mathbb T}(\bm\xi_i|\bm\xi_{-i})]=-\mathrm{MV}_0(\bm\xi_i|\bm\xi_{-i}).
$$
Therefore, mean-variance optimization and CARA utility maximization are equivalent when performed over the class of \emph{deterministic} strategies. Next, suppose that the strategies $\bm\xi_{-i}$ are deterministic. Then it follows as in Theorem 2.1 of~\cite{SchiedSchoenebornTehranchi} that the maximizer of the functional $U_\gamma(\bm\xi|\bm\xi_{-i})$ over the class of all \emph{adapted} strategies $\bm\xi\in\mathscr X(X_i,\mathbb T)$ is deterministic. That is, the maximization of $U_\gamma(\bm\xi|\bm\xi_{-i})$  over $\mathscr X(X_i,\mathbb T)$ is equivalent to the maximization over $\mathscr X_{\mathrm{det}}(X_i,\mathbb T)$. Therefore, it now follows as in the proof of Corollary 2.1 of~\cite{SchiedZhangCARA} that the strategies \eqref{opt_strat} form a Nash equilibrium for CARA utility maximization.\end{proof}

\begin{proof}[Proof of Proposition~\ref{prop:w_rho}] Let us define a vector $\bm\omega$ by $\omega_i=1$ for $i=0,\dots, N-1$ and $\omega_N=1/(1-e^{-\rho/N})$. Then the assertion will follow if
\begin{equation}
\label{pf:wrho:eq1}
(\Gamma^{0,\theta}-\widetilde{\Gamma}) \bm{\omega} = c \bm{1}
\end{equation}
with $c = \frac{1}{1-e^{-\frac{\rho}{N}}}$. To this end, we note that, with $\delta_{ij}$ denoting the Kronecker delta, 
\[
\begin{split}
\Big( (\Gamma^{0,\theta}-\widetilde{\Gamma})\bm{\omega}\Big)_i
& = \Big( (\widetilde{\Gamma}^{\top}+2\theta \text{Id})\bm{\omega}\Big)_i \\
& = (\widetilde{\Gamma}^{\top}\bm{\omega})_i + (2\theta \text{Id}\bm{\omega})_i \\
& = \frac{G(0)}{2}\omega_i + \sum^{N-1}_{j = i+1}G\Big(\frac{j}{N} - \frac{i}{N}\Big)\omega_j + G\Big(1-\frac{i}{N}\Big)\omega_N + 2\theta \sum^N_{j = 0}\delta_{ij}\omega_j \\
& = \Big(\frac{1}{2} + 2\theta\Big)\omega_i + \sum^{N-1}_{j = i+1}G\Big(\frac{j}{N} - \frac{i}{N}\Big)\omega_j + G\Big(1-\frac{i}{N}\Big)\omega_N.
\end{split}
\]
Since $\theta = \frac{1}{4}$, we have
\[
\begin{split}
\Big( (\Gamma^{0,\theta}-\widetilde{\Gamma})\bm{\omega}\Big)_i 
& = \sum^{N-1}_{j = i}G\Big(\frac{j}{N} - \frac{i}{N}\Big)\omega_j + G\Big(1-\frac{i}{N}\Big)\omega_N \\
& = e^{ \rho \frac{i}{N}}\sum^{N-1}_{j = i} e^{- \rho \frac{j}{N}} + \frac{e^{- \rho(1-\frac{i}{N})}}{1-e^{-\rho/N}} \\
&=\frac{1}{1-e^{-\rho/N}}.\end{split}
\]
This proves \eqref{pf:wrho:eq1} and hence the assertion.
\end{proof}

Now we prepare for the proof of Theorem~\ref{thm:infv}.

\begin{lemma}\label{alpha lemma}
For $\gamma,\sigma,\rho>0$, the following equation has a unique positive solution $\alpha$,
\begin{equation}
\label{eq:infv2}
\frac{1}{e^{(\alpha + \rho)}-1}-\frac{n}{e^{(\alpha - \rho)}-1} - \frac{\gamma\sigma^2 e^{-\alpha}}{(1-e^{-\alpha})^2} = 0.
\end{equation}
Moreover, $\alpha\in(0,\rho)$.
\end{lemma}

\begin{proof}By rearranging  equation  \eqref{eq:infv2} we get
\begin{align}
0&=\frac{-(e^{\rho}-e^{-\rho})+(n-1)(e^{-\alpha}-e^{\rho})}{(e^{\alpha}+e^{-\alpha})-(e^{\rho}+e^{-\rho})} +\frac{\gamma\sigma^2}{2-(e^{\alpha}+e^{-\alpha})}\nonumber\\
&=\frac{\gamma\sigma^2}{2-2\cosh(\alpha)}+\frac{-2\sinh(\rho)+(n-1)(\cosh(\alpha)-\sinh(\alpha)-(\cosh(\rho)+\sinh(\rho)))}{2\cosh(\alpha)-2\cosh(\rho)}\nonumber\\
&=\frac{\gamma\sigma^2}{2-2\cosh(\alpha)}-\frac{\sinh(\rho)}{\cosh(\alpha)-\cosh(\rho)}+\Big(\frac{n-1}{2}\Big)\Big[1-\frac{\sinh(\alpha)+\sinh(\rho)}{\cosh(\alpha)-\cosh(\rho)}\Big] =:f(\alpha)\label{pf:infv:eq2}
\end{align}
Clearly, when $\alpha > \rho>0$, then $f(\alpha)<0$. Therefore, if a zero of $f$ exists, it must be within $(0,\rho)$.
One  easily sees that
\[
\lim_{\alpha \downarrow 0}f(\alpha) = -\infty \qquad \text{and} \qquad \lim_{\alpha \uparrow \rho}f(\alpha) = +\infty.
\] 
Hence, $f$ admits at least one zero in $(0,\rho)$. Moreover, 
\[
\begin{split}
\frac{df(\alpha)}{d\alpha} 
 = & \frac{2 \gamma \sigma^2 \sinh(\alpha)}{(2-2\cosh(\alpha))^2} + \frac{\sinh(\alpha)\sinh(\rho)}{(\cosh(\alpha)-\cosh(\rho))^2} \\
& + \Big(\frac{n-1}{2}\Big) \Big[\frac{\sinh(\alpha)(\sinh(\alpha)+\sinh(\rho))}{(\cosh(\alpha)-\cosh(\rho))^2} - \frac{\cosh(\alpha)}{\cosh(\alpha) - \cosh(\rho)} \Big] > 0,
\end{split}
\]
and so the zero must be unique.
\end{proof}

\begin{lemma}\label{beta lemma}Suppose that $\gamma,\sigma,\rho>0$ and $\theta\ge0$. Then the following equation has a unique positive solution $\beta$,
\begin{equation}\label{eq:infw2} 
2\theta + \frac{1}{2} + \frac{1}{e^{(\beta+\rho)}-1}-\frac{\gamma\sigma^2 e^{-\beta}}{(1-e^{-\beta})^2} = 0.
\end{equation} 
\end{lemma}

\begin{proof}Let 
\[
g(\beta) = 2\theta + \frac{1}{2} + \frac{1}{e^{(\beta+\rho)}-1}-\frac{\gamma\sigma^2 e^{-\beta}}{(1-e^{-\beta})^2}.
\]
Clearly,
\begin{equation}
\label{pf:infw:eq2}
\lim_{\beta \downarrow 0} g(\beta) = -\infty \qquad \text{and} \qquad \lim_{\beta \uparrow \infty} g(\beta) = 2\theta + \frac{1}{2} > 0.
\end{equation}
Hence, there exists at least one zero in $(0,\infty)$. Next, we look at
\[
\begin{split}
\frac{dg(\beta)}{d\beta} 
& = -\frac{e^{(\beta+\rho)}}{(e^{(\beta+\rho)}-1)^2} + \frac{\gamma\sigma^2 e^{-\beta}}{(1-e^{-\beta})^2} + \frac{2\gamma\sigma^2 e^{-2\beta}}{(1-e^{-\beta})^3} \\
& = \gamma\sigma^2 e^{(\beta+\rho)} (e^{(\beta+\rho)}-1)^{-2} \bigg[\frac{(e^{(\beta+\rho)}-1)^2 (e^{-\beta}+1)}{e^{(2\beta+\rho)} (1-e^{-\beta})^3} - \frac{1}{\gamma\sigma^2}\bigg].
\end{split}
\]
Note that $\frac{dg(\beta)}{d\beta} > 0$ if
\[
\frac{(e^{(\beta+\rho)}-1)^2 (e^{-\beta}+1)}{e^{(2\beta+\rho)} (1-e^{-\beta})^3} - \frac{1}{\gamma\sigma^2} >0.
\]
Here, $\frac{(e^{(\beta+\rho)}-1)^2 (e^{-\beta}+1)}{e^{(2\beta+\rho)} (1-e^{-\beta})^3}$ is strictly decreasing and bounded below by $e^{\rho}$ because for all $\beta > 0$,
\[
\frac{d}{d\beta} \bigg(\frac{(e^{(\beta+\rho)}-1)^2 (e^{-\beta}+1)}{e^{(2\beta+\rho)} (1-e^{-\beta})^3} \bigg) = - \frac{2e^{(\beta - \rho)}(e^{(\beta+\rho)}-1)(2e^{(\beta+\rho)}-e^{\beta}+e^{\rho}-2)}{(e^{\beta}-1)^4} < 0,
\]
and because by L'H\^{o}pital's Rule,
\[
\begin{split}
\lim_{\beta \uparrow\infty} \bigg(\frac{(e^{(\beta+\rho)}-1)^2 (e^{-\beta}+1)}{e^{(2\beta+\rho)} (1-e^{-\beta})^3} \bigg) 
& = \lim_{\beta \uparrow\infty} \frac{(e^{(\beta+\rho)}-1)^2}{e^{(2\beta+\rho)}}  = \lim_{\beta \uparrow\infty} \frac{2e^{(\beta+\rho)}(e^{(\beta+\rho)}-1)}{2e^{(2\beta+\rho)}} \\
& = \lim_{\beta \uparrow\infty} e^{\rho} - e^{-\beta}  = e^{\rho}.
\end{split}
\]
Therefore, we have two cases to consider:
\begin{enumerate}
\item when $\gamma\sigma^2 \geq e^{-\rho}$, we have that $\frac{dg(\beta)}{d\beta} > 0$ for all $\beta > 0$. It follows that there exists a unique $\beta > 0$ such that $g(\beta) = 0$;
\item when $0 < \gamma\sigma^2 < e^{-\rho}$, we can always find a unique $\beta^*$ such that $\frac{dg(\beta)}{d\beta} > 0$ for $\beta < \beta^*$ and $\frac{dg(\beta)}{d\beta} < 0$ for $\beta > \beta^*$. In other words, $(\beta^*, g(\beta^*))$ is the global maximum of $g(\beta)$ on $(0, \infty)$. Now suppose $g(\beta^*) \leq 0$. We know that $g(\beta)$ is strictly decreasing on $(\beta^*, \infty)$, then for all $\beta \in (\beta^*, \infty)$, $g(\beta) < g(\beta^*) \leq 0$, which contradicts \eqref{pf:infw:eq2}. Therefore, $g(\beta^*) > 0$, and $0 < 2\theta + \frac{1}{2} < g(\beta) < g(\beta^*)$ for all $\beta \in (\beta^*, \infty)$, which implies a zero cannot exist in $(\beta^*, \infty)$. Since $g(\beta)$ is strictly increasing on $(0, \beta^*)$, it follows that there exists a unique $\beta \in (0, \beta^*)$ such that $g(\beta) = 0$. 
\end{enumerate}
In either case, there exists a unique positive solution $\beta$ that solves \eqref{eq:infw2}.
\end{proof}

\begin{proof}[Proof of Theorem~\ref{thm:infv}] Let us first extend the matrices \eqref{Gamma} and \eqref{Gamma_tilde} to our infinite time horizon by letting 
\begin{align*}
\Gamma_{ij}^{\gamma,\theta} &= e^{-\rho|i - j|}+ \gamma \sigma^2(i \wedge j)+2\theta\delta_{ij}, \qquad i,j=0,1,\dots,
\end{align*}
and
\begin{equation*}
\widetilde{\Gamma}_{ij} = 
\left\{
\begin{array}{ll}
      0 & \text{if } i < j, \\
      \frac12\Gamma_{ii}^{0,0} & \text{if } i = j, \\
     \Gamma_{ij}^{0,0}& \text{if } i > j. \\
\end{array} 
\right.
\end{equation*}

Next, let  $\alpha$ be as provided by Lemma~\ref{alpha lemma} and define a vector $\bm{\nu}\in\ell^1$ by
$$\nu_0 = \frac{1}{1-e^{(\alpha - \rho)}},\qquad\text{and, for $i=1,2,\dots$,}\qquad \nu_i=e^{-\alpha i}.
$$
We will now show that 
\begin{equation}
\label{pf:infv:eq1}
[\Gamma^{\gamma,\theta} + (n-1) \widetilde{\Gamma}]\bm{\nu} = c\bm{1}\qquad\text{for}\qquad c = \frac{\gamma\sigma^2 e^{-\alpha}}{(1-e^{-\alpha})^2} > 0,
\end{equation}
where $\bm 1\in\ell^\infty$ is the sequence $(1,1,\dots)$. 
Using our assumption  $\theta = \theta^* = \frac{n-1}{4}$, we get that
\begin{align*}
\lefteqn{\Big([\Gamma^{\gamma,\theta} + (n-1) \widetilde{\Gamma}]\bm{\nu}\Big)_i }\\
& = {\Gamma}^{\gamma,0}_{i0}\nu_0 + \sum^{\infty}_{j=1}{\Gamma}^{\gamma,0}_{ij}\nu_j + [2 \theta \delta_{i0} + (n-1) \widetilde{\Gamma}_{i0}]\nu_0 + \sum^{\infty}_{j=1}[2 \theta \delta_{ij} + (n-1) \widetilde{\Gamma}_{ij}]\nu_j\\
& = e^{-\rho i}\nu_0 + \sum^{\infty}_{j=1}[e^{-\rho|i-j|}+ \gamma\sigma^2(i \wedge j)]e^{-\alpha j}  + (n-1)e^{-\rho i}\nu_0 + (n-1)\sum^{i}_{j=1}e^{-\rho(i-j)}e^{-\alpha j} \\
& = ne^{-\rho i}\nu_0 + \sum^{\infty}_{j=1}[e^{-\rho|i-j|}+ \gamma\sigma^2(i \wedge j)]e^{-\alpha j} + (n-1)e^{-\rho i}\sum^{i}_{j=1} e^{-(\alpha-\rho)j}.
\end{align*}
Expanding the center term gives,
\[
\begin{split}
\sum^{\infty}_{j=1}[e^{-\rho|i-j|}+ \gamma\sigma^2(i \wedge j)]e^{-\alpha j}
& = \sum^{i}_{j=1}[e^{-\rho(i-j)}+ \gamma\sigma^2 j]e^{-\alpha j} + \sum^{\infty}_{j=i+1}[e^{-\rho(j-i)}+ \gamma\sigma^2 i]e^{-\alpha j} \\
& = e^{-\rho i}\sum^{i}_{j=1} e^{-(\alpha-\rho)j} + \gamma\sigma^2 \sum^{i}_{j=1}j e^{-\alpha j} + e^{\rho i} \sum^{\infty}_{j=i+1} e^{-(\alpha + \rho)j} + \gamma\sigma^2 i \sum^{\infty}_{j=i+1} e^{-\alpha j}.
\end{split}
\]
Thus,
\begin{align*}
\lefteqn{\Big([\Gamma^{\gamma,\theta} + (n-1) \widetilde{\Gamma}]\bm{\nu}\Big)_i}\\ &= ne^{-\rho i}\nu_0 + n e^{-\rho i}\sum^{i}_{j=1} e^{-(\alpha-\rho)j} + e^{\rho i} \sum^{\infty}_{j=i+1} e^{-(\alpha + \rho)j} + \gamma\sigma^2 \sum^{i}_{j=1}j e^{-\alpha j} + \gamma\sigma^2 i \sum^{\infty}_{j=i+1} e^{-\alpha j}\\
& = ne^{-\rho i}\nu_0+ n \Big[\frac{e^{-\rho i} - e^{-\alpha i}}{e^{(\alpha - \rho)}-1} \Big] + \frac{e^{-\alpha i}}{e^{(\alpha+\rho)}-1} + (1- e^{-\alpha i})\frac{\gamma\sigma^2 e^{-\alpha}}{(1-e^{-\alpha})^2} \\
& = n e^{-\rho i} \Big[ \frac{1}{1-e^{(\alpha-\rho)}} + \frac{1}{e^{(\alpha - \rho)}-1}\Big]  + e^{-\alpha i} \Big[\frac{1}{e^{(\alpha+\rho)}-1} -\frac{n}{e^{(\alpha - \rho)}-1} - \frac{\gamma\sigma^2 e^{-\alpha}}{(1-e^{-\alpha})^2} \Big]  + \frac{\gamma\sigma^2 e^{-\alpha}}{(1-e^{-\alpha})^2}\\
&=\frac{\gamma\sigma^2 e^{-\alpha}}{(1-e^{-\alpha})^2},
\end{align*}
where we have used $\nu_0 = \frac{1}{1-e^{(\alpha-\rho)}}$ and our equation \eqref{eq:infv2} in the final step. This establishes \eqref{pf:infv:eq1}. Now we can define 
$$\bm v=\frac1{\bm 1^\top\bm\nu}\bm\nu=\frac1{\frac1{e^\alpha-1}+\frac1{1-e^{\alpha-\rho}}}\bm \nu,
$$
which satisfies the equivalent of \eqref{eq:v} in our setting with infinite time horizon.

Let us now deal with the vector $\bm w$. To this end, we take $\beta$ as provided by Lemma~\ref{beta lemma} and define $\bm\omega\in\ell^1$ by $\omega_i=e^{-\beta i}$. Then
\begin{align*}
[(\Gamma^{\gamma,\theta}-\widetilde{\Gamma})\bm{\omega}]_i &= (2\theta + \frac{1}{2})e^{-\beta i} + e^{\rho i}\sum^{\infty}_{j=i+1}e^{-(\beta + \rho)j} + \gamma \sigma^2 \sum^i_{j=0}je^{-\beta j} + \gamma \sigma^2 i \sum^{\infty}_{j=i+1}e^{-\beta j}\\
& = e^{-\beta i} \Big[2\theta + \frac{1}{2} + \frac{1}{e^{(\beta+\rho)}-1}-\frac{\gamma \sigma^2 e^{-\beta}}{(1-e^{-\beta})^2} \Big] + \frac{\gamma \sigma^2 e^{-\beta}}{(1-e^{-\beta})^2}\\
&=\frac{\gamma \sigma^2 e^{-\beta}}{(1-e^{-\beta})^2}.
\end{align*}
It follows that we can define 
$$\bm w=\frac1{\bm1^\top\bm \omega}\bm\omega=\frac{e^\beta-1}{e^\beta}\bm\omega,
$$
 which satisfies the equivalent of \eqref{eq:w} in our setting with infinite time horizon.

 Finally, if initial positions $X_1,\dots, X_n\in\mathbb R$ are given, and we define $\bm\xi_1,\dots,\bm\xi_n$ via \eqref{opt_strat_infinity}, then it is straightforward to verify that of the first-order conditions
 \eqref{pf:optimalityconditions} are verified with our current choices for $\bm v$, $\bm w$, $\Gamma^{\gamma,\theta}$, and $\widetilde\Gamma$. As in \eqref{optimality check eq}, these yield that $\bm\xi_1,\dots,\bm\xi_n$ form a Nash equilibrium for mean-variance optimization. As in the proof of Theorem~\ref{maintheorem}, one then concludes that this is also a Nash equilibrium for CARA utility maximization.
\end{proof}


\begin{proof}[Proof of Proposition~\ref{prop:infv}]
If $n=1$ and $\gamma > 0$, we want to find an $\alpha$ that satisfies 
\[
\frac{1}{e^{(\alpha+\rho)}-1} -\frac{1}{e^{(\alpha - \rho)}-1} - \frac{\gamma \sigma^2 e^{-\alpha}}{(1-e^{-\alpha})^2} =0.
\]
If follows from \eqref{pf:infv:eq2} that
\[
\frac{\gamma \sigma^2}{2-2\cosh(\alpha)}- \frac{\sinh(\rho)}{\cosh(\alpha)-\cosh(\rho)} = 0.
\]
By rearranging the equation we have
\[
\cosh(\alpha) = \frac{\gamma \sigma^2 \cosh(\rho)+2\sinh(\rho)}{\gamma \sigma^2+2\sinh(\rho)}.
\]
Since $\rho > 0$, we have $\sinh(\rho) > 0$ and $\cosh(\rho) > 1$, which implies,
\[
\frac{\gamma \sigma^2 \cosh(\rho)+2\sinh(\rho)}{\gamma \sigma^2 +2\sinh(\rho)} > 1.
\]
Therefore, we can solve for $\alpha$ and obtain
\[
\alpha = \cosh^{-1}\Big[\frac{\gamma \sigma^2 \cosh(\rho)+2\sinh(\rho)}{\gamma \sigma^2 +2\sinh(\rho)}\Big].
\]
Furthermore, since $\rho > 0 $ implies that 
$$ \frac{\gamma \sigma^2 \cosh(\rho)+2\sinh(\rho)}{\gamma \sigma^2 +2\sinh(\rho)} < \cosh(\rho)$$
 and $\cosh^{-1}(\cdot)$ is an increasing function, we have that $0 < \alpha < \rho$. 
\end{proof}



\bibliographystyle{abbrv}
\bibliography{MarketImpact}

\begin{thebibliography}{10}

\bibitem{AFS1}
A.~Alfonsi, A.~Fruth, and A.~Schied.
\newblock Constrained portfolio liquidation in a limit order book model.
\newblock {\em Banach Center Publications}, 83:9--25, 2008.

\bibitem{ASS}
A.~Alfonsi, A.~Schied, and A.~Slynko.
\newblock Order book resilience, price manipulation, and the positive portfolio
  problem.
\newblock {\em SIAM J. Financial Math.}, 3:511--533, 2012.

\bibitem{AlmgrenChriss2}
R.~Almgren and N.~Chriss.
\newblock Optimal execution of portfolio transactions.
\newblock {\em Journal of Risk}, 3:5--39, 2000.

\bibitem{Bertsekas}
D.~P. Bertsekas.
\newblock {\em Nonlinear programming}.
\newblock Athena Scientific Optimization and Computation Series. Athena
  Scientific, Belmont, MA, second edition, 1999.

\bibitem{BertsimasLo}
D.~Bertsimas and A.~Lo.
\newblock Optimal control of execution costs.
\newblock {\em Journal of Financial Markets}, 1:1--50, 1998.

\bibitem{BrunnermeierPedersen}
M.~K. Brunnermeier and L.~H. Pedersen.
\newblock Predatory trading.
\newblock {\em Journal of Finance}, 60(4):1825--1863, August 2005.

\bibitem{Carlinetal}
B.~I. Carlin, M.~S. Lobo, and S.~Viswanathan.
\newblock Episodic liquidity crises: cooperative and predatory trading.
\newblock {\em Journal of Finance}, 65:2235--2274, 2007.

\bibitem{Carmonabook}
R.~Carmona.
\newblock {\em Lectures on {BSDE}s, stochastic control, and stochastic
  differential games with financial applications}, volume~1 of {\em Financial
  Mathematics}.
\newblock Society for Industrial and Applied Mathematics (SIAM), Philadelphia,
  PA, 2016.

\bibitem{CarmonaYang}
R.~A. Carmona and J.~Yang.
\newblock Predatory trading: a game on volatility and liquidity.
\newblock {\em Preprint}, 2011.

\bibitem{CasgrainJaimungal}
P.~Casgrain and S.~Jaimungal.
\newblock Algorithmic trading with partial information: A mean field game
  approach.
\newblock {\em arXiv:1803.04094}, 2018.

\bibitem{SEC}
CFTC-SEC.
\newblock Findings regarding the market events of {M}ay 6, 2010.
\newblock Report, 2010.

\bibitem{Gatheral}
J.~Gatheral.
\newblock No-dynamic-arbitrage and market impact.
\newblock {\em Quant. Finance}, 10:749--759, 2010.

\bibitem{GatheralSchiedSurvey}
J.~Gatheral and A.~Schied.
\newblock Dynamical models of market impact and algorithms for order execution.
\newblock In J.-P. Fouque and J.~Langsam, editors, {\em Handbook on Systemic
  Risk}, pages 579--602. Cambridge University Press, 2013.

\bibitem{HubermanStanzl}
G.~Huberman and W.~Stanzl.
\newblock Price manipulation and quasi-arbitrage.
\newblock {\em Econometrica}, 72(4):1247--1275, 07 2004.

\bibitem{LorenzAlmgren}
J.~Lorenz and R.~Almgren.
\newblock Mean-variance optimal adaptive execution.
\newblock {\em Appl. Math. Finance}, 18(5):395--422, 2011.

\bibitem{Moallemietal}
C.~C. Moallemi, B.~Park, and B.~Van~Roy.
\newblock Strategic execution in the presence of an uninformed arbitrageur.
\newblock {\em Journal of Financial Markets}, 15(4):361 -- 391, 2012.

\bibitem{ow}
A.~Obizhaeva and J.~Wang.
\newblock Optimal trading strategy and supply/demand dynamics.
\newblock {\em Journal of Financial Markets}, 16:1--32, 2013.

\bibitem{Polya49}
G.~P\'olya.
\newblock Remarks on characteristic functions.
\newblock In J.~Neyman, editor, {\em Proceedings of the Berkeley Symposium of
  Mathematical Statistics and Probability}, pages 115--123. University of
  California Press, 1949.

\bibitem{SchiedSchoenebornTehranchi}
A.~Schied, T.~Sch\"oneborn, and M.~Tehranchi.
\newblock Optimal basket liquidation for {CARA} investors is deterministic.
\newblock {\em Applied Mathematical Finance}, 17:471--489, 2010.

\bibitem{SchiedStrehleZhang}
A.~Schied, E.~Strehle, and T.~Zhang.
\newblock High-frequency limit of {N}ash equilibria in a market impact game
  with transient price impact.
\newblock {\em SIAM J. Financial Math.}, 8(1):589--634, 2017.

\bibitem{SchiedZhangCARA}
A.~Schied and T.~Zhang.
\newblock A state-constrained differential game arising in optimal portfolio
  liquidation.
\newblock {\em Math. Finance}, 27(3):779--802, 2017.

\bibitem{SchiedZhangHotPotato}
A.~Schied and T.~Zhang.
\newblock A market impact game under transient price impact.
\newblock {\em Mathematics of Operations Research}, 44(1):102--121, 2019.

\bibitem{Schoeneborn}
T.~Sch\"oneborn.
\newblock Trade execution in illiquid markets. {O}ptimal stochastic control and
  multi-agent equilibria.
\newblock {D}octoral dissertation, TU Berlin, 2008.

\bibitem{SchoenebornSchied}
T.~Sch\"oneborn and A.~Schied.
\newblock Liquidation in the face of adversity: stealth vs. sunshine trading.
\newblock SSRN Preprint 1007014, 2009.

\bibitem{Strehle}
E.~Strehle.
\newblock Optimal execution in a multiplayer model of transient price impact.
\newblock {\em Market Microstructure and Liquidity}, 3(4):1850007, 2017.

\bibitem{Young}
W.~H. Young.
\newblock On the {F}ourier series of bounded functions.
\newblock {\em Proceedings of the London Mathematical Society (2)}, 12:41--70,
  1913.

\end{thebibliography}

\end{document}